\documentclass[pra,aps,showpacs,twocolumn,twoside,superscriptaddress]{revtex4}

\usepackage{amsmath,amsfonts,amssymb,caption,color,epsfig,graphics,graphicx,hyperref,latexsym,mathrsfs,revsymb,theorem,url,verbatim,epstopdf}

\hypersetup{colorlinks,linkcolor={blue},citecolor={red},urlcolor={blue}}

\newtheorem{definition}{Definition}
\newtheorem{proposition}[definition]{Proposition}
\newtheorem{lemma}[definition]{Lemma}

\newtheorem{theorem}[definition]{Theorem}
\newtheorem{corollary}[definition]{Corollary}
\newtheorem{conjecture}[definition]{Conjecture}

\newtheorem{remark}[definition]{Remark}
\newtheorem{example}[definition]{Example}
\newtheorem{question}[definition]{Question}

\def\bcj{\begin{conjecture}}
\def\ecj{\end{conjecture}}
\def\bcr{\begin{corollary}}
\def\ecr{\end{corollary}}
\def\bd{\begin{definition}}
\def\ed{\end{definition}}
\def\bea{\begin{eqnarray}}
\def\eea{\end{eqnarray}}
\def\bem{\begin{enumerate}}
\def\eem{\end{enumerate}}
\def\bex{\begin{example}}
\def\eex{\end{example}}
\def\bim{\begin{itemize}}
\def\eim{\end{itemize}}
\def\bl{\begin{lemma}}
\def\el{\end{lemma}}
\def\bma{\begin{bmatrix}}
\def\ema{\end{bmatrix}}
\def\bpf{\begin{proof}}
\def\epf{\end{proof}}
\def\bpp{\begin{proposition}}
\def\epp{\end{proposition}}
\def\bqu{\begin{question}}
\def\equ{\end{question}}
\def\br{\begin{remark}}
\def\er{\end{remark}}
\def\bt{\begin{theorem}}
\def\et{\end{theorem}}

\def\squareforqed{\hbox{\rlap{$\sqcap$}$\sqcup$}}
\def\qed{\ifmmode\squareforqed\else{\unskip\nobreak\hfil
\penalty50\hskip1em\null\nobreak\hfil\squareforqed
\parfillskip=0pt\finalhyphendemerits=0\endgraf}\fi}
\def\endenv{\ifmmode\;\else{\unskip\nobreak\hfil
\penalty50\hskip1em\null\nobreak\hfil\;
\parfillskip=0pt\finalhyphendemerits=0\endgraf}\fi}
\newenvironment{proof}{\noindent \textbf{{Proof.~} }}{\qed}
\def\Dbar{\leavevmode\lower.6ex\hbox to 0pt
{\hskip-.23ex\accent"16\hss}D}
\makeatletter
\def\url@leostyle{%
  \@ifundefined{selectfont}{\def\UrlFont{\sf}}{\def\UrlFont{\small\ttfamily}}}
\makeatother
\def\bcj{\begin{conjecture}}
\def\ecj{\end{conjecture}}
\def\bcr{\begin{corollary}}
\def\ecr{\end{corollary}}
\def\bd{\begin{definition}}
\def\ed{\end{definition}}
\def\bea{\begin{eqnarray}}
\def\eea{\end{eqnarray}}
\def\bem{\begin{enumerate}}
\def\eem{\end{enumerate}}
\def\bex{\begin{example}}
\def\eex{\end{example}}
\def\bim{\begin{itemize}}
\def\eim{\end{itemize}}
\def\bl{\begin{lemma}}
\def\el{\end{lemma}}
\def\bpf{\begin{proof}}
\def\epf{\end{proof}}
\def\bpp{\begin{proposition}}
\def\epp{\end{proposition}}
\def\bqu{\begin{question}}
\def\equ{\end{question}}
\def\br{\begin{remark}}
\def\er{\end{remark}}
\def\bt{\begin{theorem}}
\def\et{\end{theorem}}

\def\btb{\begin{tabular}}
\def\etb{\end{tabular}}

\newcommand{\nc}{\newcommand}

\def\l{\lambda}
\def\m{\mu}

\def\s{\sigma}

 \nc{\bbA}{\mathbb{A}} \nc{\bbB}{\mathbb{B}} \nc{\bbC}{\mathbb{C}}
 \nc{\bbD}{\mathbb{D}} \nc{\bbE}{\mathbb{E}} \nc{\bbF}{\mathbb{F}}
 \nc{\bbG}{\mathbb{G}} \nc{\bbH}{\mathbb{H}} \nc{\bbI}{\mathbb{I}}
 \nc{\bbJ}{\mathbb{J}} \nc{\bbK}{\mathbb{K}} \nc{\bbL}{\mathbb{L}}
 \nc{\bbM}{\mathbb{M}} \nc{\bbN}{\mathbb{N}} \nc{\bbO}{\mathbb{O}}
 \nc{\bbP}{\mathbb{P}} \nc{\bbQ}{\mathbb{Q}} \nc{\bbR}{\mathbb{R}}
 \nc{\bbS}{\mathbb{S}} \nc{\bbT}{\mathbb{T}} \nc{\bbU}{\mathbb{U}}
 \nc{\bbV}{\mathbb{V}} \nc{\bbW}{\mathbb{W}} \nc{\bbX}{\mathbb{X}}
 \nc{\bbZ}{\mathbb{Z}}

 \nc{\bA}{{\bf A}} \nc{\bB}{{\bf B}} \nc{\bC}{{\bf C}}
 \nc{\bD}{{\bf D}} \nc{\bE}{{\bf E}} \nc{\bF}{{\bf F}}
 \nc{\bG}{{\bf G}} \nc{\bH}{{\bf H}} \nc{\bI}{{\bf I}}
 \nc{\bJ}{{\bf J}} \nc{\bK}{{\bf K}} \nc{\bL}{{\bf L}}
 \nc{\bM}{{\bf M}} \nc{\bN}{{\bf N}} \nc{\bO}{{\bf O}}
 \nc{\bP}{{\bf P}} \nc{\bQ}{{\bf Q}} \nc{\bR}{{\bf R}}
 \nc{\bS}{{\bf S}} \nc{\bT}{{\bf T}} \nc{\bU}{{\bf U}}
 \nc{\bV}{{\bf V}} \nc{\bW}{{\bf W}} \nc{\bX}{{\bf X}}
 \nc{\bZ}{{\bf Z}}

\nc{\cA}{{\cal A}} \nc{\cB}{{\cal B}} \nc{\cC}{{\cal C}}
\nc{\cD}{{\cal D}} \nc{\cE}{{\cal E}} \nc{\cF}{{\cal F}}
\nc{\cG}{{\cal G}} \nc{\cH}{{\cal H}} \nc{\cI}{{\cal I}}
\nc{\cJ}{{\cal J}} \nc{\cK}{{\cal K}} \nc{\cL}{{\cal L}}
\nc{\cM}{{\cal M}} \nc{\cN}{{\cal N}} \nc{\cO}{{\cal O}}
\nc{\cP}{{\cal P}} \nc{\cQ}{{\cal Q}} \nc{\cR}{{\cal R}}
\nc{\cS}{{\cal S}} \nc{\cT}{{\cal T}} \nc{\cU}{{\cal U}}
\nc{\cV}{{\cal V}} \nc{\cW}{{\cal W}} \nc{\cX}{{\cal X}}
\nc{\cZ}{{\cal Z}}

\nc{\hA}{{\hat{A}}} \nc{\hB}{{\hat{B}}} \nc{\hC}{{\hat{C}}}
\nc{\hD}{{\hat{D}}} \nc{\hE}{{\hat{E}}} \nc{\hF}{{\hat{F}}}
\nc{\hG}{{\hat{G}}} \nc{\hH}{{\hat{H}}} \nc{\hI}{{\hat{I}}}
\nc{\hJ}{{\hat{J}}} \nc{\hK}{{\hat{K}}} \nc{\hL}{{\hat{L}}}
\nc{\hM}{{\hat{M}}} \nc{\hN}{{\hat{N}}} \nc{\hO}{{\hat{O}}}
\nc{\hP}{{\hat{P}}} \nc{\hR}{{\hat{R}}} \nc{\hS}{{\hat{S}}}
\nc{\hT}{{\hat{T}}} \nc{\hU}{{\hat{U}}} \nc{\hV}{{\hat{V}}}
\nc{\hW}{{\hat{W}}} \nc{\hX}{{\hat{X}}} \nc{\hZ}{{\hat{Z}}}

\nc{\hn}{{\hat{n}}}

\def\diag{\mathop{\rm diag}}

\def\lin{\mathop{\rm span}}

\def\sr{\mathop{\rm sr}}

\def\dg{\dagger}

\newcommand{\ket}[1]{|#1\rangle}
\newcommand{\proj}[1]{| #1\rangle\!\langle #1 |}
\newcommand{\ketbra}[2]{|#1\rangle\!\langle#2|}

\def\Dbar{\leavevmode\lower.6ex\hbox to 0pt
{\hskip-.23ex\accent"16\hss}D}

\begin{document}
\title{Constructing three-qubit unitary gates in terms of Schmidt rank and CNOT gates}

\date{\today}


\author{Zhiwei Song}
\affiliation{School of Mathematical Sciences, Beihang University, Beijing 100191, China}
\affiliation{Department of Applied Mechanics,  University of Science and Technology Beijing, Beijing 100083, China}

\author{Lin Chen}\email[]{linchen@buaa.edu.cn (corresponding author)}
\affiliation{School of Mathematical Sciences, Beihang University, Beijing 100191, China}
\affiliation{International Research Institute for Multidisciplinary Science, Beihang University, Beijing 100191, China}

\author{Mengyao Hu}\email[]{mengyaohu@buaa.edu.cn (corresponding author)}
\affiliation{School of Mathematical Sciences, Beihang University, Beijing 100191, China}

\begin{abstract}
It is known that every two-qubit unitary operation has Schmidt rank one, two or four, and the construction of three-qubit unitary gates in terms of Schmidt rank remains an open problem. We explicitly construct the gates of Schmidt rank from one to seven. It turns out that the three-qubit Toffoli and Fredkin gate respectively have Schmidt rank two and four. As an application, we implement the gates using quantum circuits of CNOT gates and local Hadamard and flip gates. In particular, the collective use of three CNOT gates can generate a three-qubit unitary gate of Schmidt rank seven in terms of the known Strassen tensor from multiplicative complexity. Our results imply the connection between the number of CNOT gates for implementing multiqubit gates and their Schmidt rank.
\end{abstract}

\maketitle



\section{Introduction}

The implementation of multiqubit unitary gates is one of the central problems in quantum computing \cite{PhysRevA.52.3457,Chau1995Simple,Smolin1996Five,Yu2015Optimal}.  It has been shown that every two-qubit unitary operation has Schmidt rank one, two or four \cite{Nielsen03}. 
The Schmidt rank plays a key role when determining whether a bipartite unitary operation is a controlled unitary operation \cite{cy13,cy14,cy14ap}, the decomposition of multipartite unitary gates into the product of controlled unitary gates for implementing efficiently quantum circuits \cite{cy15}, and the derivation of entangling power of bipartite unitaries for quantifying how much entanglement they can create locally \cite{cy16,cy16b}. 

As far as we know, it's an open problem of characterizing multiqubit unitary operations in terms of Schmidt rank. In this paper, we construct three-qubit unitary matrices of Schmidt rank from one to seven, respectively. We introduce the preliminary fact of deriving the Schmidt rank of tripartite matrices in Lemma \ref{le:3vector=span} and Corollary \ref{cr:prod}. The construction is presented in Theorem \ref{thm:3qubit<=7}. We also present a three-qubit unitary gate of Schmidt rank seven or eight in Theorem \ref{thm:sr7or8}. This is supported by Lemma \ref{le:sr(u8)>=6}. It turns out that the well-known three-qubit Toffoli and Fredkin gate respectively have Schmidt rank two and four. Then we implement three-qubit unitary gates of Schmidt rank one to seven using controlled-NOT (CNOT) gates and local unitary gates such as the Hadamard gates and qutrit flip gates. We illustrate the implementation in Figure \ref{fig:toffoli} to \ref{fig:u7}. In Theorem \ref{thm:two CNOT}, we show that three CNOT gates are necessary for the implementation of gates of Schmidt rank three, five, six and seven. Furthermore, we show in Theorem \ref{thm:three CNOT} that the collective use of three CNOT gates can generate a three-qubit unitary gate of Schmidt rank seven in terms of the Strassen tensor from multiplicative complexity \cite{Landsberg2011Tensors}.  

The implementation of quantum gates is usually carried out using CNOT gates assisted with local unitary gates. The efficiency is thus evaluated by the number of CNOT gates involved in the implementation. It has been proved that the theoretical lower bound for the number of CNOT gates needed in simulating an arbitrary $n$-qubit gate is $\lceil\frac{1}{4}(4^n-3n-1)\rceil$ \cite{shende2004minimal,vartiainen2004efficient}. So far there is little study on the connection between the Schmidt rank of a multiqubit gate and the number of required CNOT gates. Our results thus initiate the problem of understanding quantum circuit in terms of Schmidt rank.

The rest of this paper is organized as follows. In Sec. \ref{sec:res} we introduce the preliminary knowledge of this paper. Then we construct three-qubit unitary operations of Schmidt rank one to seven, respectively. We also construct a three-qubit unitary operation of Schmidt rank seven or eight. In Sec. \ref{sec:app} we implement three-qubit unitary gates using CNOT gates assisted by local unitary gates. We conclude in Sec. \ref{sec:con}.

\section{Construction of three-qubit untiary gates}
\label{sec:res}

We begin by introducing the notations used in this paper. We refer to $\bbC^d$ as the $d$-dimensional Hilbert space. We  denote $\bbM_{a\times b}$ as the set of $a\times b$ matrices. In particular if $a=b$ then we refer to $\bbM_a$ as the set of $a\times a$ matrices. Let $M^\dg$ be the transpose and complex conjugate of matrix $M$, i.e., $M^\dg =(M^T)^*$. Let $I_n$ be the $n\times n$ identity matrix. Further we shall refer to $I_2,\s_1,\s_2$ and $\s_3$ as the identity matrix and three Pauli matrices, respectively. Further, we denote $S_1,S_2,S_3,S_4$ as the $2\times2$ matrices 
\begin{eqnarray}
\label{eq:s0123}
&&
S_0=\bma1&0\\0&0 \ema, 
\quad
S_1=\bma0&1\\0&0 \ema,
\notag\\&&
S_2=\bma0&0\\1&0 \ema, 
\quad
S_3=\bma0&0\\0&1 \ema.
\end{eqnarray}

We define the Schmidt rank of an $n$-partite matrix $U$ on the $n$-partite Hilbert space $\cH_1\otimes...\otimes\cH_n:=\bbC^{d_1}\otimes...\otimes\bbC^{d_n}$ as the minimum integer $r$ such that 
$
U=\sum^r_{j=1}
A_{j,1}\otimes...\otimes A_{j,n-1}	\otimes A_{j,n}	
$
for some $d_i \times d_i$ matrix $A_{j,i}$ and $i=1,...,n$ \footnote{The notion is equivalent to the tensor rank in matrix multiplication. We denote it as Schmidt rank because a similar use has been proposed in \cite{Briegel2001The}. }. If $n=2$ then the definition reduces to the Schmidt rank of bipartite matrix $U$. For convenience we refer to $\sr(U)$ as the Schmidt rank of $U$. One can effectively derive the Schmidt rank of bipartite matrix by computing the  rank of the matrix modified from the bipartite matrix. Unfortunately computing the Schmidt rank of a tripartite matrix is an NP-hard problem \cite{H1990Tensor}. Nevertheless, we can construct the relation between bipartite and multipartite matrices, so as to investigate the relation between the Schmidt rank of them. For example, we can regard $U$ as a bipartite unitary matrix $U_{S:\bar{S}}$ of system $S=\{1,...,k\}$ and $\bar{S}=\{k+1,...,n\}$. By writing the Schmidt decomposition of $U_{S:\bar{S}}$, i.e., $U_{S:\bar{S}}=\sum^r_{i=1}B_i \otimes C_i$ with $r=\sr(U_{S:\bar{S}})$, we shall say that the span of $B_1,...,B_r$ is the $S$-space of $U$. Similarly, the span of $C_1,...,C_r$ is the $\bar{S}$ space of $U$. It's straightforwardly to show the inequality $\sr(U)\ge \sr(U_{S:\bar{S}})$. This is a frequently used lower bound of the Schmidt rank of $U$ because the Schmidt rank of bipartite matrices are known to be computable. 
We will use the inequality in the paper without explanation unless stated otherwise.  

To find a systematic way of deriving the Schmidt rank, we review a fact from Theorem 3.1.1.1 on p68 of \cite{LandsbergTensors}.  
\begin{lemma}
\label{le:3vector=span}
Suppose	$U=\sum^r_{j=1} Q_j \otimes R_j$ is a tripartite matrix where $Q_j$ on $\cH_A \otimes \cH_B$ are linearly independent, and $R_j$ on $\cH_C$ are also linearly independent. Then the Schmidt rank of $U$ is the minimal number of product matrices spanning the space including the space spanned by $Q_1,...,Q_r$.
\qed
\end{lemma}

Then we present a corollary of this lemma.
\begin{corollary}
\label{cr:prod}
We still use the notations in Lemma 	\ref{le:3vector=span}. Let $U=\sum^{\sr(U)}_{i=1}X_i \otimes Y_i \otimes Z_i$. If $Q_1,...,Q_n$ are product matrices then we may assume that $Q_j=X_j\otimes Y_j$ for $j=1,...,n$.  
\end{corollary}
\begin{proof}
We know that $Q_i$ is the linear combination of $X_j\otimes Y_j$ for $j=1,2,..,\sr(U)$. If $i=1$ and the coefficient of $X_1\otimes Y_1$ is nonzero, then we may express $X_1\otimes Y_1$ as the linear combination of $Q_1$ and $X_2\otimes Y_2,...,X_{\sr(U)}\otimes Y_{\sr(U)}$. Using the expression we obtain that $X_j\otimes Y_j$ is the linear combination of the same matrices. Hence we may assume that $Q_1=X_1\otimes Y_1$. One can similarly prove the assertion for $j=2,...,n$.
\end{proof}

The above corollary plays an important role in constructing three-qubit unitary matrices of Schmidt rank from one to seven, respectively. This is presented in Theorem \ref{thm:3qubit<=7}, namely the first main result of this section. Next we construct the three-qubit unitary operation of Schmidt rank seven or eight in Theorem \ref{thm:sr7or8}. This is the second main result of this section. We begin by studying three-qubit unitary matrices of Schmidt rank up to seven.
\begin{theorem}
\label{thm:3qubit<=7}
The three-qubit unitary operation of Schmidt rank up to seven exists.	
\end{theorem}
\begin{proof}
Let $U$ be a three-qubit unitary operation. 	If suffices to find $U$ with $\sr(U)=1,2,3,4,5$ and $6$, respectively. It is known that the two-qubit unitary $V$ of Schmidt rank one, two or four exists. So $U=I_2\otimes V$ has Schmidt rank one, two or four.  Next one can show that $U_3={1\over\sqrt3} (I_2^{\otimes3}+i\s_1^{\otimes3}+i\s_3^{\otimes3})$ is a three-qubit unitary matrix of Schmidt rank three. 

Third we construct $U=U_5$ of Schmidt rank five. Let
\begin{eqnarray}
\label{eq:3qubit<=5}
U_5=&&
{1\over2}S_0
\otimes(I_2\otimes I_2+\s_1\otimes\s_1+\s_2\otimes\s_2+ \s_3\otimes\s_3)	
\notag\\+&& S_3\otimes I_2\otimes \s_1
\notag\\=&&
{1\over2}S_0
\otimes I_2 \otimes I_2
+
{1\over2}S_0
\otimes \s_1\otimes\s_1+
S_3
\otimes I_2\otimes\s_1
\notag\\+&&
{1\over2}S_0
\otimes \s_2\otimes\s_2+
{1\over2}S_0
\otimes \s_3\otimes\s_3.
\end{eqnarray}
One can show that $U_5$ is unitary, and $4\le \sr(U_5)\le 5$. If $\sr(U_5)=4$ then $U_5=\sum^4_{j=1}A_j\otimes B_j\otimes C_j$ with some $2\times2$ matrices $A_j,B_j$ and $C_j$. By comparing with \eqref{eq:3qubit<=5}, one can show that $C_j$'s are linear independent, namely they span the space of the $2\times2$ matrices. So the $AB$ space of $U_5$ is spanned by $A_j\otimes B_j$'s, namely four linearly independent product matrices. Using \eqref{eq:3qubit<=5}, one can show that the $AB$ space of $U_5$ is spanned by the four linearly independent matrices 
\begin{eqnarray}
\label{eq:12e11}	
&&
{1\over2}S_0\otimes \s_1+S_3\otimes I_2, 
\quad
S_0\otimes I_2, 
\notag\\&&
S_0\otimes \s_2, 
\quad
S_0\otimes \s_3.
\end{eqnarray}
The assertion at the end of last paragraph says that, each of the four linearly independent product matrices $A_j\otimes B_j$'s is the linear combination of the four matrices in \eqref{eq:12e11}. So at least one of $A_j\otimes B_j$'s is the linear combination of them such that the coefficient of ${1\over2}S_0\otimes \s_1+S_3\otimes I_2$ is nonzero. However one can show that this linear combination is not a product matrix. 
We have proven that $\sr(U_5)\ne4$. Hence $\sr(U_5)=5$.

Fourth we construct $U=U_6$ of Schmidt rank six. Let
\begin{eqnarray}
\label{eq:ue11}
U_6=&&
{1\over2}S_0
\otimes(I_2\otimes I_2+\s_1\otimes\s_1+\s_2\otimes\s_2+ \s_3\otimes\s_3)	
\notag\\+&&
{1\over\sqrt2}S_3\otimes 
(I_2\otimes \s_1+\s_2\otimes
\s_3)
\notag\\=&&
{1\over2}S_0
\otimes I_2 \otimes I_2
+
{1\over2}S_0
\otimes \s_1\otimes\s_1
\notag\\+&&
{1\over\sqrt2}
S_3
\otimes I_2\otimes\s_1+
{1\over2}S_0
\otimes \s_2\otimes\s_2
\notag\\+&&
{1\over2}S_0
\otimes \s_3\otimes\s_3
+
{1\over\sqrt2}
S_3
\otimes \s_2\otimes\s_3.
\end{eqnarray} 

Suppose that $\sr(U_6)\le 5$. We may assume that $U_6=\sum^5_{i=1}Q_i\otimes C_i$ with the product matrices $Q_i\in \bbM_2\otimes\bbM_2$. Using Corollary 
\ref{cr:prod} and \eqref{eq:ue11} we may assume that 
\begin{eqnarray}
&& Q_1=S_0 \otimes I_2,
\\	
&& Q_2=S_0 \otimes \s_2,
\\&& 
\label{eq:ajqj}
S_0\otimes\s_1+\sqrt2 S_3\otimes I_2
=\sum^5_{j=1} a_jQ_j,
\\&& 
\label{eq:bjqj}
S_0\otimes\s_3+\sqrt2 S_3\otimes \s_2
=\sum^5_{j=1} b_jQ_j,
\end{eqnarray}
for some complex numbers $a_j$ and $b_j$. Let $Q_j=A_j\otimes B_j$ with $2\times2$ matrices $A_j$ and $B_j$ for $j=3,4,5$. Eqs. \eqref{eq:ajqj} and \eqref{eq:bjqj} imply that $\s_1,I_2,\s_3,\s_2\in\lin\{B_3,B_4,B_5\}$. It is a contradiction with the fact that $\lin\{B_3,B_4,B_5\}$ has dimension at most three. We have shown that $\sr(U_6)\ge6$. On the other hand \eqref{eq:ue11} shows that $\sr(U_6)\le6$. Hence $\sr(U_6)=6$. 

Fifth we construct $U=U_7$ of Schmidt rank seven. Let 
\begin{eqnarray}
U_7=&&
S_1\otimes S_2\otimes S_0+
S_2\otimes S_3\otimes S_0
\notag\\+&&
S_0\otimes S_0\otimes S_1+
S_3\otimes S_1\otimes S_1
\notag\\+&&
S_1\otimes S_1\otimes S_2+
S_2\otimes S_0\otimes S_2
\notag\\+&&
S_0\otimes S_3\otimes S_3+
S_3\otimes S_2\otimes S_3,
\end{eqnarray}

One can verify that $U_7$ is unitary. Further, we perform the permutation $(3210)$ on system $A$, $(320)$ on system $B$, and $(13)$ on system $C$ of $U_7$. Then $U_7$ is isomorphic to the known $4\times4\times4$ Strassen tensor, which has Schmidt rank seven. Hence $\sr(U_7)=7$. We have proven the assertion. 
\end{proof}

In contrast to the gate of Schmidt rank four constructed in the above proof, one can show that the three-qubit unitary operation in Eq. (18) of the paper \cite{Bullock2003Canonical}, written as $U=
{1\over\sqrt2}(S_0\otimes I_2\otimes I_2+
S_1\otimes \s_3\otimes \s_3+
S_2\otimes \s_1\otimes \s_1+
S_3\otimes \s_2\otimes \s_2)
$, has also Schmidt rank four. It is the so-called finagler related to the standard Cartan involution. Furthermore, one can show that the four-qubit unitary $U'$ in Eq. (16) of the paper \cite{Bullock2003Canonical} has rank at most 16. Actually we can express $U'$ as the sum of $16$ product matrices as follows.

\begin{widetext}
\begin{eqnarray}
\notag U'=&&
{1\over\sqrt2}(S_0\otimes S_0\otimes S_0
+S_0\otimes S_1\otimes S_2+
S_1\otimes S_2\otimes S_0+
S_1\otimes S_3\otimes S_2)\otimes \bma1&i\\0&0 \ema\\
\notag +&&
(S_0\otimes S_0\otimes S_1
+S_0\otimes S_1\otimes S_3+
S_1\otimes S_2\otimes S_1+
S_1\otimes S_3\otimes S_3)\otimes \bma0&0\\1&i \ema\\
\notag +&&
(S_2\otimes S_2\otimes S_2 
-S_2\otimes S_3\otimes S_0
-S_3\otimes S_0\otimes S_2+
S_3\otimes S_1\otimes S_0)\otimes \bma0&0\\1&-i \ema\\
+&&
(S_2\otimes S_2\otimes S_3 
-S_2\otimes S_3\otimes S_1-
S_3\otimes S_0\otimes S_3+
S_3\otimes S_1\otimes S_1)\otimes \bma-1&i\\0&0 \ema.
\end{eqnarray}
\end{widetext}

In the following, we construct a three-qubit unitary matrix $U_8$ using \eqref{eq:s0123}, and show it has Schmidt rank seven or eight in Theorem \ref{thm:sr7or8}. This is the second main result of this section. 

\begin{eqnarray}
\label{eq:u8}
U_8:=&&
S_0\otimes S_0\otimes S_0+
S_1\otimes S_3\otimes S_0
\notag\\+&&
S_2\otimes S_0\otimes S_1+
S_3\otimes S_2\otimes S_1
\notag\\+&&
S_0\otimes S_1\otimes S_2+
S_1\otimes S_2\otimes S_2
\notag\\+&&
S_2\otimes S_1\otimes S_3+
S_3\otimes S_3\otimes S_3.
\end{eqnarray}
We present the following observation as a lower bound of Schmidt rank of $U_8$.
\begin{lemma}
\label{le:sr(u8)>=6}	
The Schmidt rank of tensor $S_1\otimes S_3\otimes S_0+
S_2\otimes S_0\otimes S_1+
S_3\otimes S_2\otimes S_1+
S_1\otimes S_2\otimes S_2+
S_2\otimes S_1\otimes S_3+
S_3\otimes S_3\otimes S_3$ is six. Furthermore $\sr(U_8)\ge6$. 
\end{lemma}
\begin{proof}
Let $T=S_1\otimes S_3\otimes S_0+
S_2\otimes S_0\otimes S_1+
S_3\otimes S_2\otimes S_1
+
S_1\otimes S_2\otimes S_2+
S_2\otimes S_1\otimes S_3+
S_3\otimes S_3\otimes S_3$. Because $U_8$ can be projected onto $T$ using a projector on the first system, we obtain that $\sr(T)\le \sr(U_8)$. So it suffices to prove $\sr(T)=6$ by contradiction. Suppose $\sr(T)\le5$, namely
$
T=\sum^5_{j=1}A_j\otimes B_j\otimes C_j.	
$ Using the orthogonality of $S_0,S_1,S_2,S_3$ we obtain that
$
S_2\otimes S_0+S_3\otimes S_2,
S_2\otimes S_1+S_3\otimes S_3, 
S_1\otimes S_3, S_1\otimes S_2 
\in\lin\{A_1\otimes B_1,...,A_5\otimes B_5 \}.	
$
By setting $A_1\otimes B_1=S_1\otimes S_3$ and $A_2\otimes B_2=S_1\otimes S_2$, we obtain that $S_0,S_1,S_2,S_3\in\lin\{B_3,B_4, B_5\}$. It is a contradiction, so we have shown that $\sr(T)=6$.	
\end{proof}

Now we are in a position to present the second main result of this section.
\begin{theorem}
\label{thm:sr7or8}
$\sr(U_8)=7$ or $8$. 	
\end{theorem}
\begin{proof}
Suppose $\sr(U_8)=6$, we have 
\begin{eqnarray}
\label{eq:u8f6}
U_8:=&&
F_1\otimes G_1\otimes H_1+F_2\otimes G_2\otimes H_2+F_3\otimes G_3\otimes H_3
\notag\\+&&
F_4\otimes G_4\otimes H_4+F_5\otimes G_5\otimes H_5+F_6\otimes G_6\otimes H_6,
\notag\\
\end{eqnarray}
where $F_i,G_i,H_i$ are $2\times2$ matrices. $S_0$ is orthogonal to $S_1,S_2 ,S_3$ implies that $S_0\otimes S_0+S_1\otimes S_2$ is in the span of $G_1\otimes H_1$,...$G_6\otimes H_6$. So we assume that 
\begin{eqnarray}
\label{eq:u8f12}
U_8:=&&
F_7\otimes (S_0\otimes S_0+S_1\otimes S_2)+F_8\otimes G_2\otimes H_2
\notag\\+&&
F_9\otimes G_3\otimes H_3+F_{10}\otimes G_4\otimes H_4
\notag\\+&&
F_{11}\otimes G_5\otimes H_5+F_{12}\otimes G_6\otimes H_6.
\end{eqnarray}
Futher, we get $S_0-F_7$ is in the span of $S_1,S_2,S_3$. Let $S_0-F_7=xS_1+yS_2+zS_3$ where $x,y,z$ are complex numbers and at least one of them are nonzero. We have
\begin{eqnarray}
\label{eq:u8f12}
U_8:=&&
S_1\otimes (S_3\otimes S_0+S_2\otimes S_2+x(S_0\otimes S_0+S_1\otimes S_2))
\notag\\=&&S_2\otimes (S_0\otimes S_1+S_1\otimes S_3+y(S_0\otimes S_0+S_1\otimes S_2))
\notag\\=&&S_3\otimes (S_2\otimes S_1+S_3\otimes S_3+z(S_0\otimes S_0+S_1\otimes S_2))
\notag\\=&&
F_8\otimes G_2\otimes H_2+F_9\otimes G_3\otimes H_3+F_{10}\otimes G_4\otimes H_4
\notag\\+&&
F_{11}\otimes G_5\otimes H_5+F_{12}\otimes G_6\otimes H_6.
\end{eqnarray}
Note that Schmidt rank is invariant under invertible local transformation, we do the transformation $S_3\rightarrow S_3-xS_0$, $S_2\rightarrow S_2-xS_1$ on system $B$ and $S_1\rightarrow S_1-yS_0$, $S_3\rightarrow S_3-yS_2 $ on system $C$, and obtain 
\begin{eqnarray}
\label{eq:u8f121}
&&S_1\otimes (S_3\otimes S_0+S_2\otimes S_2)+S_2\otimes (S_0\otimes S_1+S_1\otimes S_3)
\notag\\+&&
S_3\otimes ((S_2-xS_1)\otimes (S_1-yS_0)+z(S_0\otimes S_0+S_1\otimes S_2)
\notag\\+&&
(S_3-xS_0)\otimes (S_3-yS_2))
\notag\\=&&
F_{8}'\otimes G_{2}'\otimes H_{2}'+F_{9}'\otimes G_{3}'\otimes H_{3}'+F_{10}'\otimes G_{4}'\otimes H_{4}'
\notag\\+&&
F_{11}'\otimes G_{5}'\otimes H_{5}'+F_{12}'\otimes G_{6}'\otimes H_{6}'.
\end{eqnarray}
We next define $n_{s_1}$ as the number of matrices in the set $\left\lbrace F_{8}',F_{9}',F_{10}',F_{11}',F_{12}'\right\rbrace$ that are not orthogonal to $S_1$. Eq.(\ref{eq:u8f121}) implies that $2\le n_{s_1}\le 5$. 

We shall investigate $n_{s_1}$ in four cases. First, suppose $n_{s_1}=2$. Up to the switch of $F_i'\otimes G_i'\otimes H_i'$, we can assume that $F_{8}'$ and $F_{9}'$ are not orthogoanl to $S_1$ and hence $S_3\otimes S_0+S_2\otimes S_2=p G_{2}'\otimes H_{2}'+q G_{3}'\otimes H_{3}'$ for nonzero complex numbers $p,q$. Hence $H_{2}',H_{3}'\in\lin\{S_0,S_2\}$. Futher, by regarding $S_0,S_1,S_2,S_3$ as 4-dim vectors and performing the local projection $I_A\otimes I_B \otimes (\ketbra{S_1}{S_1}+\ketbra{S_3}{S_3})$ on Eq.(\ref{eq:u8f121}), we obtain that 
\begin{eqnarray}
&&S_2\otimes (S_0\otimes S_1+S_1\otimes S_3)+S_3\otimes ((S_2-xS_1)\otimes S_1
\notag \\+&&
(S_3-xS_0)\otimes S_3)
\notag \\=&&
F_{10}'\otimes G_{4}'\otimes H_{4}'+F_{11}'\otimes G_{5}'\otimes H_{5}'+F_{12}'\otimes G_{6}'\otimes H_{6}'.
\notag\\
\end{eqnarray}
However, the equtaion does not hold because the left has Schmidt rank four while the right has three entries at most. So $n_{s_1}\ne 2$.

The other three cases are $n_{s_1}=3,4$ or $5$. And for all, we can use the similar way to prove that they are impossible. Hence $\sr(U_8)=6$ is impossible.

Using Lemma \ref{le:sr(u8)>=6}, we finish the proof.
\end{proof}

Unfortunately we cannot determine $\sr(U)=7$ or $8$, and we leave it as an open problem. In the next section, we shall show how to construct some three-qubit unitary operations of Schmidt rank from one to seven.

\section{Implementation of three-qubit unitary gates}
\label{sec:app}

We have shown in Theorem \ref{thm:3qubit<=7} the existence of three-qubit unitary gates of Schmidt rank one to seven. It is a natural question to ask how many CNOT gates $T:=\proj{0}\otimes I_2+\proj{1}\otimes \s_1$ are sufficient to implement them. In this section we investigate the question. To save CNOT gates, we will construct three-qubit gates of various Schmidt rank different from those in Theorem \ref{thm:3qubit<=7}. In particular, we show that the three-qubit Toffoli and Fredkin gate respectively have Schmidt rank two and four. We show in Theorem \ref{thm:two CNOT} that
the combination of two CNOT gates and local unitary gates generate a three-qubit unitary gate of Schmidt rank one, two or four. So implementing gates of Schmidt rank three and larger than four requires at least three CNOT gates. In particular, we show in Theorem \ref{thm:three CNOT} that the combination of three CNOT gates can generate a three-qubit unitary gate of Schmidt rank seven, by using the isomorphism to the Strassen tensor from multiplicative complexity. 

First, every Schmidt-rank-one unitary gate is a local unitary gate, and it does not require CNOT gate. Next, the three-qubit gate $U_2=I_A\otimes T_{BC}$ has Schmidt rank two, and can be implemented using one CNOT gate. As it is trivial, we construct a nontrivial example. We point out that the known three-qubit Toffoli gate $T_3$ (i.e., the controlled CNOT gate) also has Schmidt rank two, because 
\begin{eqnarray}
T_3=&&
(I_2\otimes I_2\otimes H)
(I_2\otimes I_2\otimes I_2
\notag\\-&&
2\proj{1}\otimes \proj{1}\otimes \proj{1}) 	
(I_2\otimes I_2\otimes H),
\end{eqnarray}
where  $H=\bma {\sqrt2\over 2}&{\sqrt2\over 2}\\{\sqrt2\over 2}&-{\sqrt2\over 2} \ema$ stands for the qubit Hadamard gate. It has been proven that the Toffoli gate can be implemented using three CNOT gates assisted with local gates  \cite{PhysRevA.75.022313, Lanyon2008Simplifying}, see Figure \ref{fig:toffoli}.

\begin{figure}[ht]
\centering
\includegraphics[width=8cm]{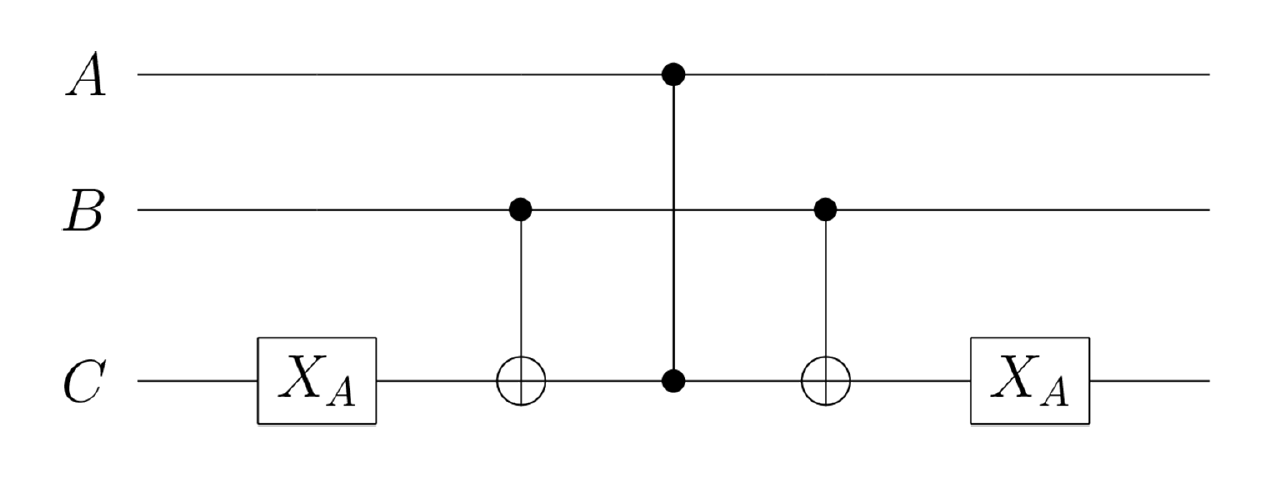}
\caption{The three-qubit Toffoli gate $T_3$ of Schmidt rank two can be realized using two CNOT gates and one CZ gate $\diag(1,1,1,-1)$ in the middle. The CZ gate is locally equivalent to the CNOT gate via two Hadamard gates $H$. The local gate $X_A$ flips the qutrits $\ket{0}$ and $\ket{2}$.}
\label{fig:toffoli}
\end{figure}

Third using the Toffoli gate and one more CNOT gate, we can construct a three-qubit gate $U_3$ of Schmidt rank three as follows.
\begin{eqnarray}
U_3=&&
(T_{AB}\otimes H)T_3
(I_2\otimes I_2\otimes H)
\notag\\=&&	
(\proj{0}\otimes I_2+\proj{1}\otimes \s_1)
\otimes I_2
\notag\\-&& 
2\proj{1}\otimes \ketbra{0}{1}\otimes \proj{1}. 
\end{eqnarray}
So we can realize $U_3$ using four CNOT gates assisted with local unitary gates in Figure \ref{fig:sr3}.

\begin{figure}[ht]
\centering
\includegraphics[width=8cm]{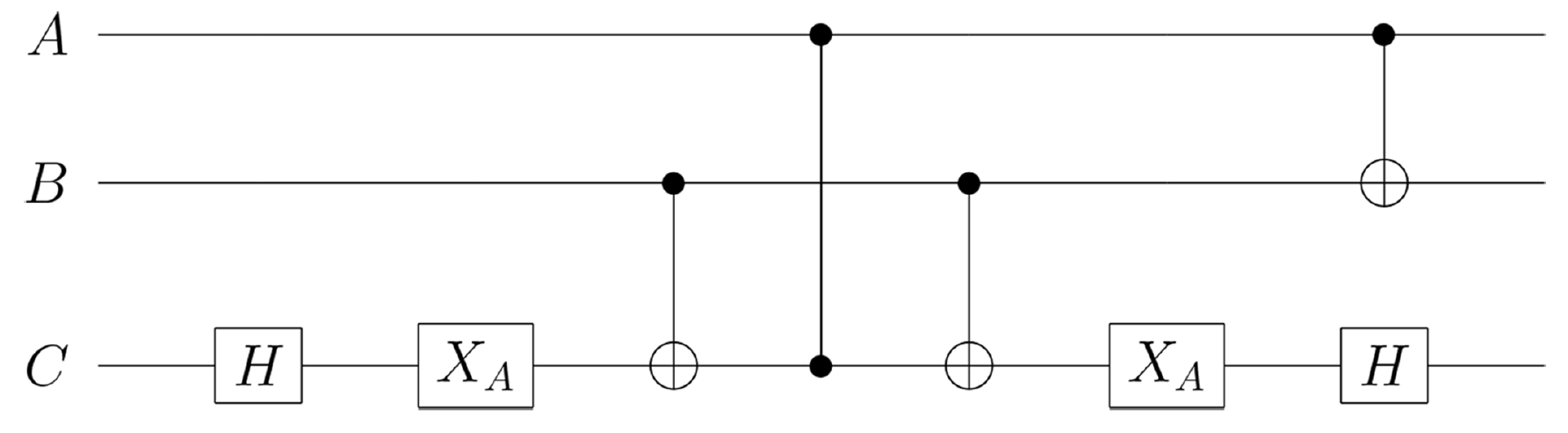}
\caption{The three-qubit gate $U_3$ of Schmidt rank three can be implemented using four CNOT gates, local Hadamard gates $H$ and local gate $X_A$ flipping the qutrit $\ket{0}$ and $\ket{2}$.}
\label{fig:sr3}
\end{figure}

We don't know whether four CNOT gates are also necessary for constructing a thre-qubit unitary gate of Schmidt rank three. Nevertheless, It turns out that three CNOT gates are necessary. This is a corollary of the following observation. 
\begin{theorem}
\label {thm:two CNOT}
The combination of two CNOT gates and local unitary gates generates a three-qubit unitary gate of Schmidt rank one, two or four. 
\end{theorem}
\begin{proof}
Up to the switch of systems, the combination of two CNOT gates and local unitary gates has the expression $M_1=U_1((T_{AB}\otimes I_C)U(I_A\otimes T_{BC}))U_2$ or $M_2=U_1((T_{AB}\otimes I_C)U(T_{AB}\otimes I_C))U_2$, with local three-qubit unitary gates $U_1,U$ and $U_2$. One can verify that the second gate $M_2$ is indeed a two-qubit unitary gate, so it does not have Schmidt rank three. By choosing $U=I_8$, the gate $M_2$ becomes a local unitary gate. 

Next we consider $M_1$, suppose $U=V \otimes W\otimes X$, and it also has the expression $M_1=U_1(I_2\otimes I_2\otimes X)(T_{AB}\otimes I_C)(I_2\otimes W\otimes I_2)(I_A\otimes T_{BC})(V\otimes I_2\otimes I_2)U_2$. Because local unitary transformation does not change the Schmidt rank, we may assume that $U_1=U_2=I_8$ and $V=X=I_2$. We have

\begin{eqnarray}
\notag M_1=&&(S_0\otimes I_2\otimes I_2+S_3\otimes \sigma_1\otimes I_2)
\notag\\&&
(I_2 \otimes W\otimes I_2)(I_2\otimes S_0\otimes I_2+I_2\otimes S_3\otimes \sigma_1)
\notag\\
\notag =&&S_0\otimes WS_0\otimes I_2+S_0\otimes WS_3\otimes \sigma_1\\
  +&&S_3\otimes \s_1 WS_0\otimes I_2+S_3\otimes \s_1 WS_3\otimes \sigma_1,
\end{eqnarray}	
where $S_0,...,S_3$ are the $2\times2$ matrices defined in \eqref{eq:s0123}.
So $S_1$ also has Schmidt rank at most four. It is clear that $S_0$ and $S_3$ are linearly independent in system $A$, $I_2$ and $\sigma_1$ are linearly independent in system $C$. We next consider the four matrices $WS_0,WS_3, \s_1WS_0$ and $\s_1WS_3$ in system $B$.

Assume that $k_1WS_0+k_2\s_1WS_0+k_3WS_3+k_4\s_1WS_3=0$ for complex numbers $k_1$ to $k_4$ and set $W=\bma m&n\\l&p \ema$.

We obtain that
\begin{eqnarray}
\bma mk_1+lk_2&nk_3+pk_4\\lk_1+mk_2&pk_3+nk_4 \ema=\bma 0&0\\0&0 \ema,
\end{eqnarray} 
and
\begin{eqnarray}
\label{m l} mk_1+lk_2=0,\\
\label{l m} lk_1+mk_2=0,\\
nk_3+pk_4=0,\\
pk_3+nk_4=0.
\end{eqnarray}

Suppose $k_1=0$, if $k_2=0$ then from Eqs. \eqref{m l} and \eqref{l m} we have $m=n=0$, it means that $U$ cannot be a unitary matrix, so this is impossible. Hence $k_1=0$ implies $k_2=0$. Also $k_2=0$ implies $k_1=0$ and the same relations to $k_3$ and $k_4$.

We next suppose $k_1 \ne 0$ and hence $k_2 \ne 0$, we obtain 
\begin{eqnarray}
{k_1 \over k_2}={m\over l}={l\over m}. 
\end{eqnarray}

So $m^2=l^2$ and hence $n^2=p^2$. We can get the same result if we assume $k_3\ne 0$ and $k_4\ne 0$.  

In both cases we have $WS_0=\bma m&0\\l&0 \ema$ and $\s_1WS_0=\bma l&0\\m&0 \ema$ are linearly dependent, $WS_3=\bma 0&n\\0&p \ema$ and $\s_1WS_3=\bma 0&p\\0&n \ema$ are linearly dependent, and obtain
$
 M_1=(S_0+{m\over l}S_3)\otimes WS_0\otimes I_2+(S_0+{p\over n}S_3)\otimes WS_3\otimes \s_1.
$
So in this situation $M_1$ has Schmidt rank two. 

The remaining case is that $k_1=k_2=k_3=k_4=0$. So the four matrices in system $B$ are linearly independent. Further, any three product matrices could not span the AB space of $M_1$. So $M_1$ has Schmidt rank four.
We finish the proof.
\end{proof}

Fourth, we construct the three-qubit unitary gate $U_4=(T_{AB}\otimes I_C)(I_A\otimes T_{BC})$. It is straightforward to prove that $U_4$ has Schmidt rank four. We describe it in Figure \ref{fig:sr4}. Note that two CNOT gates are the minimum cost of realizing every gate of Schmidt rank four. In contrast, we point out that the known three-qubit Fredkin gate $F_3$ (i.e., the controlled swap gate) also has Schmidt rank four, because 
\begin{eqnarray}
F_3=&&
(\proj{0}+{1\over2}\proj{1})
\otimes I_2\otimes 
I_2
+{1\over2}\proj{1}
\otimes\s_3\otimes\s_3
\notag\\+&&
\proj{1}
\otimes
\ketbra{0}{1}
\otimes
\ketbra{1}{0}
+
\proj{1}
\otimes
\ketbra{1}{0}
\otimes
\ketbra{0}{1}.
\notag\\
\end{eqnarray}
It's been proven that the Fredkin gate can be implemented using five CNOT gates assisted with local gates \cite{2004.03134}, see Figure \ref{fig:fredkin}. We don't know whether the Fredkin gate can be implemented using fewer CNOT gates.

\begin{figure}[ht]
\centering
\includegraphics[width=5cm]{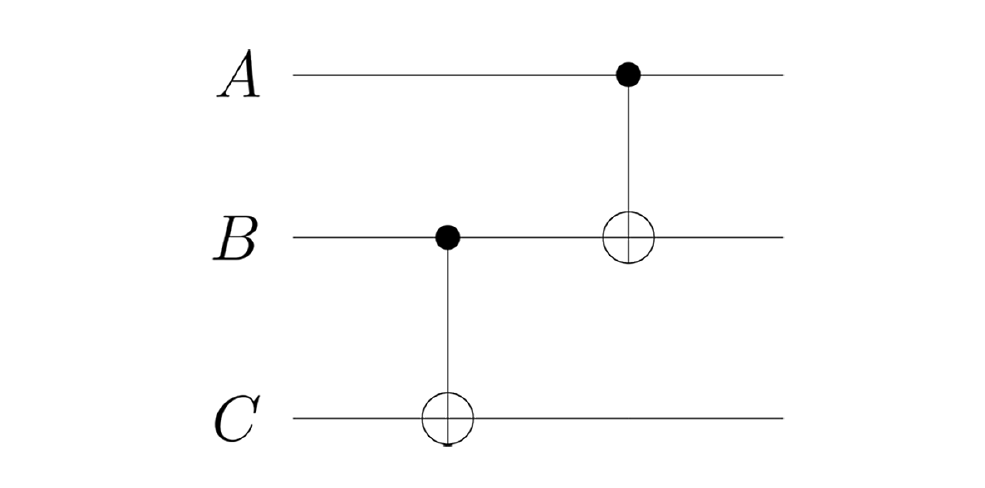}
\caption{The three-qubit gate $U_4$ of Schmidt rank four consists of two CNOT gates. This is minimum cost of realizing any three-qubit unitary gate of Schmidt rank four.}
\label{fig:sr4}
\end{figure}

\begin{figure}[ht]
\centering
\includegraphics[width=8cm]{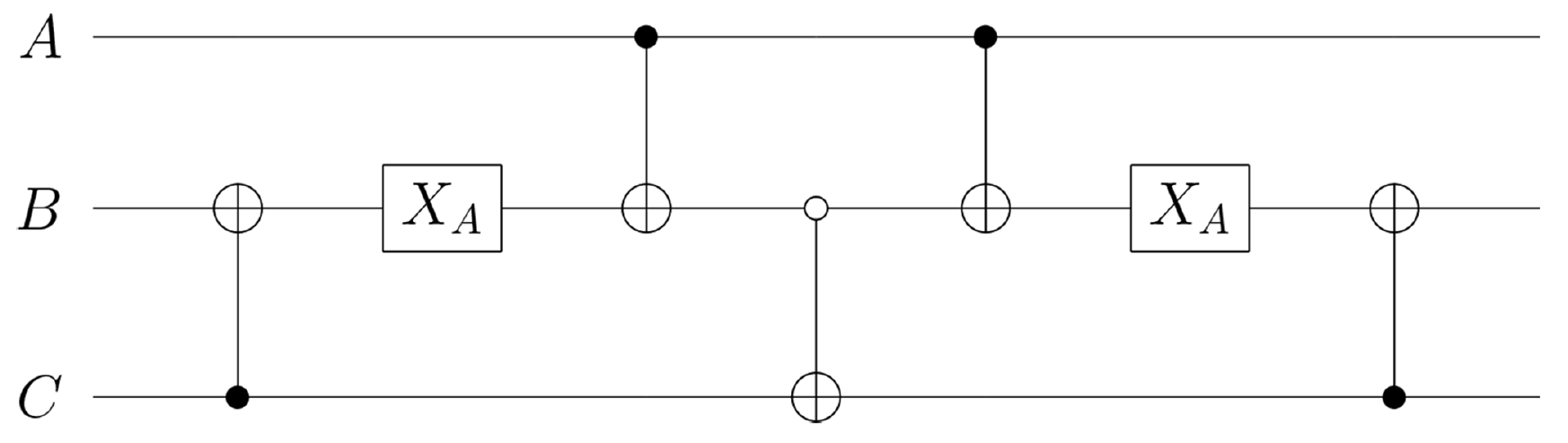}
\caption{The three-qubit Fredkin gate $F_3$ of Schmidt rank four can be implemented using five CNOT gates and local gates $X_A$ flipping the qutrit $\ket{0}$ and $\ket{2}$.}
\label{fig:fredkin}
\end{figure}


Fifth, using the Fredkin gate  and one more CNOT gate, we can construct a three-qubit gate $U_5$ of Schmidt rank five as follows.
\begin{eqnarray}
\label{eq:sr5}
U_5=&&
(T_{AB}\otimes I_2)F_3
=\proj{0}
\otimes I_2\otimes 
I_2
\notag\\	
+&&
{1\over2}\proj{1}\otimes \s_1\otimes 
I_2
+
{1\over2}\proj{1}
\otimes \s_1\s_3\otimes\s_3
\notag\\+&&
\proj{1}
\otimes
\ketbra{1}{1}
\otimes
\ketbra{1}{0}
+
\proj{1}
\otimes
\ketbra{0}{0}
\otimes
\ketbra{0}{1}.
\notag\\
\end{eqnarray}
We explain briefly why $\sr(U_5)=5$, as the  proof is similar to that of constructing the gate in \eqref{eq:3qubit<=5}. First using \eqref{eq:sr5} one can show that $5\ge\sr(U_5)\ge4$. Next if $\sr(U)=4$ then $U_5$ is the linear combination of four product matrices one of which has the form $A \otimes I_2\otimes I_2$. It can be excluded by comparing with \eqref{eq:sr5}. We have shown that $\sr(U)=5$. Using Figure \ref{fig:fredkin}, we can implement $U_5$ using six CNOT gates assisted with local unitary gates in Figure
\ref{fig:sr5}. 

\begin{figure}[ht]
\centering
\includegraphics[width=9cm]{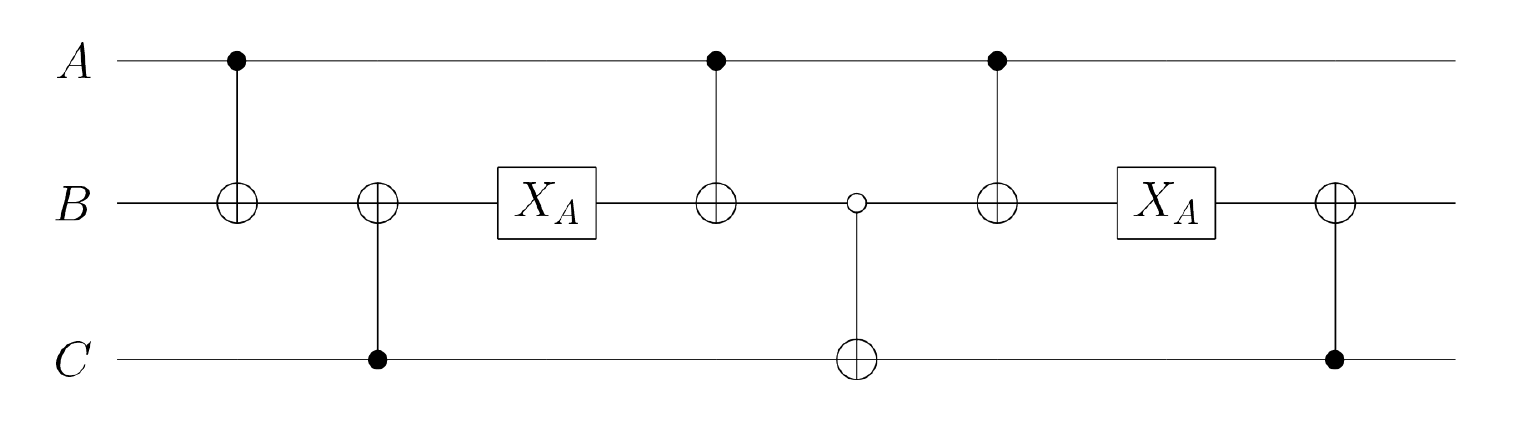}
\caption{The three-qubit gate $U_5$ of Schmidt rank five can be implemented using six CNOT gates and local gates $X_A$ flipping the qutrit $\ket{0}$ and $\ket{2}$.}
\label{fig:sr5}
\end{figure}

Sixth, using the gate $U_3$ in Figure \ref{fig:sr3} and one more CNOT gate, we can construct a three-qubit gate $U_6$ of Schmidt rank six as follows.
\begin{eqnarray}
U_6=&&
(T_{AC}\otimes (I_2)_B)
(H\otimes I_2\otimes I_2 )U_3
\notag\\=&&	
{1\over\sqrt2}
\ketbra{0}{0}
\otimes (I_2\otimes I_2)
\notag\\+&&
{1\over\sqrt2}
\ketbra{0}{1}
\otimes
(\s_1\otimes I_2-2\ketbra{0}{1}\otimes \proj{1})
\notag\\+&&
{1\over\sqrt2}
\ketbra{1}{0}
\otimes I_2\otimes \s_1
\notag\\+&&
{1\over\sqrt2}
\ketbra{1}{1}
\otimes
(2\ketbra{0}{1}\otimes \ketbra{0}{1}-\s_1\otimes \s_1),
\end{eqnarray}
where $H=\bma {\sqrt2\over 2}&{\sqrt2\over 2}\\{\sqrt2\over 2}&-{\sqrt2\over 2} \ema$ is the Hadamard matrix. We explain briefly why $\sr(U_6)=6$, as the  proof is similar to that of constructing the gate in \eqref{eq:ue11}. First Corollary \ref{cr:prod} shows that $6\ge\sr(U_6)\ge4$. Next if $\sr(U_6)\le 5$ then one can show that $\s_1\otimes I_2-2\ketbra{0}{1}\otimes \proj{1}$ and $2\ketbra{0}{1}\otimes \ketbra{0}{1}-\s_1\otimes \s_1$ cannot be in the span of $I_2\otimes I_2$, $I_2\otimes \s_1$ and any three product matrices. We have a contradiction and so $\sr(U_6)=6$. Using Figure \ref{fig:sr3}, we can implement $U_6$ using five CNOT gates assisted with local unitary gates in Figure
\ref{fig:sr6}. 

\begin{widetext}

\begin{figure}[ht]
	\centering
	\includegraphics[width=9cm]{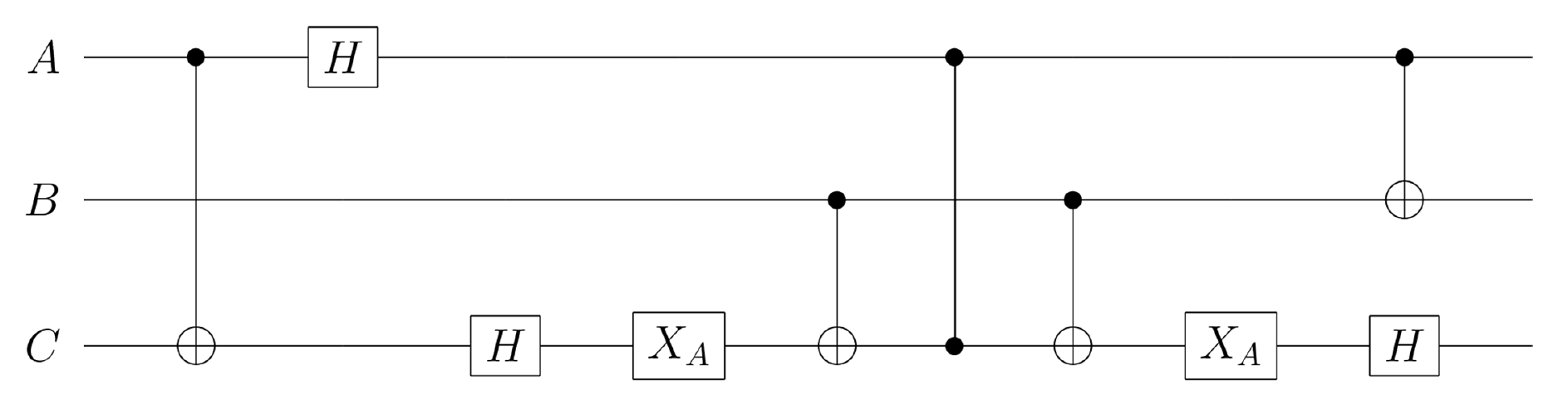}
	\caption{The three-qubit gate $U_6$ of Schmidt rank six can be implemented using five CNOT gates, Hadamard gate $H$ and qutrit gate $X_A$ flipping $\ket{0}$ and $\ket{2}$.}
	\label{fig:sr6}
\end{figure}
\end{widetext}

It remains to implement a three-qubit unitary gate of Schmidt rank seven using CNOT gates as few as possible. Fortunately this is the case by the following theorem. 
\begin{theorem}
\label{thm:three CNOT}
The combination of three CNOT gates can generate a three-qubit unitary gate of Schmidt rank seven. 	
\end{theorem}
\begin{proof}
Consider the expression $M_3=(T_{AB}\otimes I_C)U_1(I_A\otimes T_{BC})U_2(T_A\otimes I_B\otimes T_C)$ where $U_1=V_1\otimes W_1\otimes X_1$ and $U_2=V_2\otimes W_2\otimes X_2$ are local three-qubit unitary gates, set $X_1=I_2$ and $W_2=I_2$, we have
\begin{eqnarray}
\label{eq:S3}
\notag M_3=&&(S_0V_1V_2S_0\otimes W_1S_0+S_3V_1V_2S_0\otimes \s_1W_1S_0)\otimes X_2\\
\notag +&&(S_0V_1V_2S_0\otimes W_1S_3+S_3V_1V_2S_0\otimes \s_1W_1S_3)\otimes \s_1X_2\\
\notag +&&(S_0V_1V_2S_3\otimes W_1S_0+S_3V_1V_2S_3\otimes \s_1W_1S_0)\otimes X_2\s_1\\
+&&(S_0V_1V_2S_3\otimes W_1S_3+S_3V_1V_2S_3\otimes \s_1W_1S_3)\otimes \s_1X_2\s_1.
\notag\\
\end{eqnarray}
Next, assume $V_1=I_2$, using the Hadamard gate $V_2=H$, it is easy to show that the four martices $S_0V_1V_2S_0$,  $S_0V_1V_2S_3$,  $S_3V_1V_2S_0$ and  $S_3V_1V_2S_3$ in system $A$ are linearly independent. Assume $W_1=I_2$, it is easy to show that the four matrices $W_1S_0$, $W_1S_3$, $\s_1W_1S_0$ and $\s_1W_1S_3$ in system $B$ are linearly independent. Next, assume $X_2=H$ is also a Hadamard gate, and it implies the four matrices $X_2$, $\s_1X_2$, $X_2\s_1$ and $\s_1X_2\s_1$ in system $C$ are linearly independent.

Based on these conditions, we obtain that the three-qubit unitary gate
\begin{eqnarray}
\label{eq:s3=sr7}
M_3=(T_{AB}\otimes I_C)(I_A\otimes T_{BC})(T_{CA} \otimes I_B)
\end{eqnarray}
is isomorphic to the Strassen Tensor and hence it has Schmidt rank seven. We describe \eqref{eq:s3=sr7} in Figure \ref{fig:u7}.
\end{proof}

\begin{figure}[ht]
\centering
\includegraphics[width=5cm]{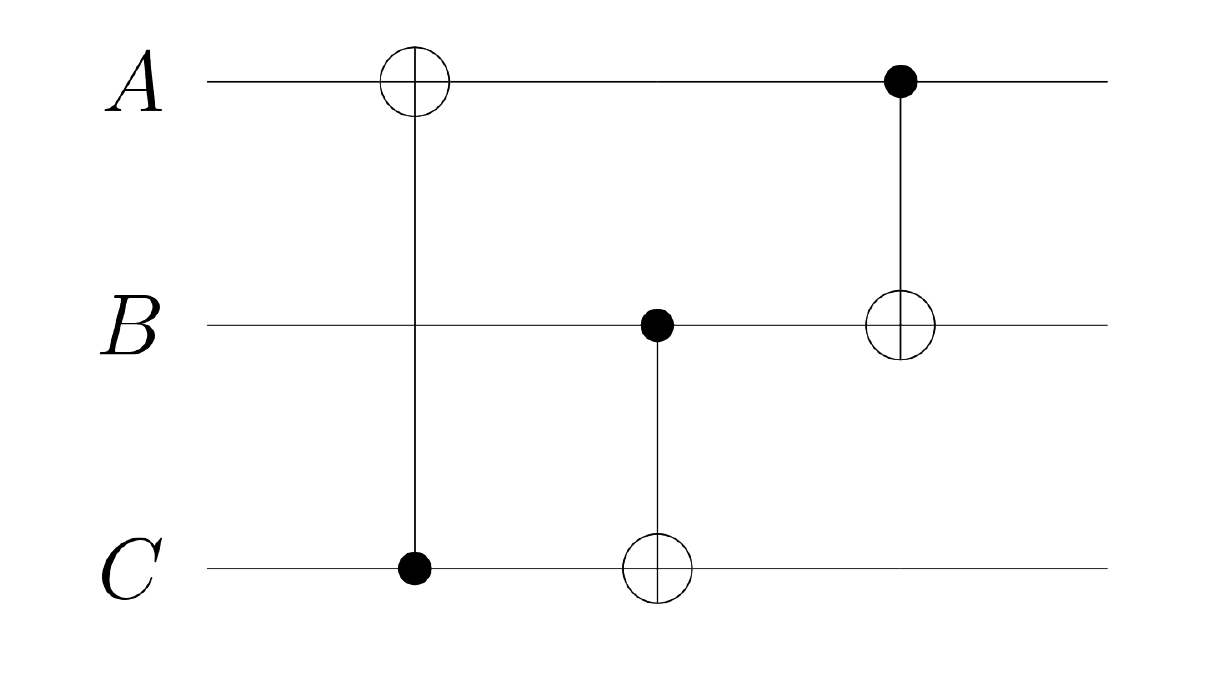}
\caption{The three-qubit gate of Schmidt rank seven consists of three CNOT gates.}
\label{fig:u7}
\end{figure}

Using Theorem \ref{thm:two CNOT}, three CNOT gates are also necessary for implementing any three-qubit unitary gate of Schmidt rank seven. On the other hand,  constructing the gates of Schmidt rank three to six in Figure \ref{fig:sr3}, \ref{fig:fredkin}, \ref{fig:sr5} and \ref{fig:sr6} costs more than three CNOT gates. It is an interesting problem to reduce the numbers or prove their necessity if possible.

\section{Conclusions}
\label{sec:con}

We have constructed three-qubit unitary operations of Schmidt rank from one to seven, respectively. We have implemented them using CNOT gates and local unitary gates. It remains to  determine whether the three-qubit unitary operations of Schmidt rank eight and nine exist, and investigate their extension to multiqubit quantum circuit.

\section*{Acknowledgments}
	\label{sec:ack}	
We thank Shmuel Friedland and Delin Chu for valuable discussions. Authors were supported by the  NNSF of China (Grant No. 11871089), and the Fundamental Research Funds for the Central Universities (Grant No. ZG216S2005).

\bibliographystyle{unsrt}

\bibliography{3qubit}

\end{document}